\documentclass[12pt]{article}

\usepackage[margin=1.2in]{geometry}
\usepackage[T1]{fontenc}
\usepackage[utf8]{inputenc}
\usepackage{lmodern}
\usepackage{microtype}
\usepackage{amsmath,amssymb,amsthm,bm}
\usepackage{booktabs}
\usepackage{graphicx}
\usepackage{float}
\usepackage{multirow}
\usepackage[colorlinks=true,citecolor=blue,linkcolor=blue,urlcolor=blue]{hyperref}
\usepackage[authoryear,round]{natbib}
\usepackage{setspace}
\usepackage{caption}
\usepackage{subcaption}
\usepackage{threeparttable}

\numberwithin{equation}{section}
\allowdisplaybreaks[2]

\theoremstyle{plain}
\newtheorem{proposition}{Proposition}[section]

\title{\large\bfseries Do Prediction Markets Match Option Prices?\\[4pt]
       Bitcoin Threshold Evidence from Binance and Polymarket}

\author{Victoria Portnaya\\[4pt]
        \normalsize Kyiv School of Economics\\
        \normalsize Kyiv, Ukraine\\[2pt]
        \normalsize \href{mailto:vportnaia@kse.org.ua}{vportnaia@kse.org.ua}}

\date{June 2026}

\begin{document}

\maketitle
\thispagestyle{empty}

\begin{abstract}
\noindent The digitization of financial markets has produced two classes of platforms
that price, in principle, the same state-contingent payoffs: centralized
crypto-option exchanges and blockchain-based prediction markets. This paper provides
the first option-implied benchmark test of prediction-market pricing for
cryptocurrency threshold contracts. For each hour in a matched sample, I compare the
Polymarket Yes price, $P_{poly,t}$, with the discounted risk-neutral binary value
implied by a listed Binance call option on the same underlying, strike, and maturity,
$P_{fair,t}$, and study the gap $D_t = P_{poly,t} - P_{fair,t}$. In the main
September~2023 Bitcoin contract, the mean pricing gap equals 5.6 percentage points
across 214 hourly observations ($t = 6.46$, $p < 10^{-9}$). Pooling three
Binance-compatible Bitcoin threshold markets yields a mean gap of 6.3 percentage
points across 287 observations, robust to HAC and block-bootstrap inference. The gap
is persistent---with an AR(1) half-life of roughly four hours---yet mean-reverting,
consistent with slow information transmission between segmented venues rather than
mechanical noise. Cross-sectional regressions reveal that the wedge is largest at
low option-implied probabilities and long maturities, a pattern consistent with
speculative demand for prediction-market contracts rather than measurement error.
A delta-hedged arbitrage proxy remains profitable after conservative transaction
costs, though with marginal statistical precision. A Deribit extension on the same
three Bitcoin contracts produces a larger pooled gap of 11 percentage points, while
a smaller Ethereum exercise yields mixed evidence. The results demonstrate that
digital fragmentation of financial markets generates systematic, persistent pricing
wedges even for economically identical payoffs.

\bigskip
\noindent\textbf{JEL codes:} G13, G14, G19

\noindent\textbf{Keywords:} prediction markets; options pricing; risk-neutral
probabilities; Bitcoin; cross-venue arbitrage; Polymarket
\end{abstract}

\newpage
\setcounter{page}{1}
\tableofcontents
\newpage

\section{Introduction}
\label{sec:intro}

The digitization of financial markets has produced two classes of platforms that
price, in principle, the same state-contingent payoffs. On one side sit centralized
crypto-option exchanges---algorithmically matched, continuously quoted, and
integrated with global derivatives infrastructure. On the other side sit
blockchain-based prediction markets---decentralized, oracle-settled, and populated
by a participant base whose access to financial markets has been transformed by
digital platforms. Whether these two digital ecosystems agree on the value of an
identical binary payoff is an open empirical question with implications for how
digitization shapes price discovery and cross-market information transmission.

Prediction markets are increasingly treated as probability aggregators, but this
interpretation requires a benchmark. For a contract that pays one dollar if Bitcoin
closes above a fixed threshold on a fixed date, the natural benchmark is not a
survey forecast or a historical frequency---it is the market price of a financial
claim with an identical payoff. This paper exploits that observation to conduct what
is, to our knowledge, the first option-implied benchmark test of prediction-market
pricing for cryptocurrency threshold contracts.

A vanilla call option traded on a centralized exchange can be inverted to recover
the market-implied probability of a threshold event under standard no-arbitrage
pricing. Comparing that probability with the prediction-market price for the same
event yields a sharp cross-market test: if both venues price the same
state-contingent payoff efficiently, their average difference should be
indistinguishable from zero after accounting for frictions. If one market is
systematically above or below the other, the difference should show up clearly in
the data.

I implement this test using exact Binance historical option data, exact Binance spot
bars, and Polymarket trade-level data for three Bitcoin threshold contracts from
mid-2023. The empirical design is deliberately narrow. Only contracts with a precise
option counterpart---identical underlying, strike, and maturity---enter the sample.
This restriction sacrifices breadth for comparability: it rules out
contract-mapping ambiguity as an alternative explanation and ensures that any
measured wedge is a statement about pricing, not about payoff structure.

The main finding is a robust positive wedge. Polymarket Yes prices were, on average,
5.6 percentage points above Binance-implied risk-neutral values in the
September~2023 contract and 6.3 percentage points above in the pooled three-market
panel. This gap is statistically precise under classical, HAC, and block-bootstrap
inference. It is economically meaningful relative to estimated frictions. And it is
persistent over intraday horizons---with an autoregressive half-life of
approximately four hours---rather than transitory noise. These features rule out
simple explanations based on bid--ask spreads or quote-staleness.

The structure of the wedge provides additional discipline. Cross-sectional
regressions on the pooled panel show that the gap is largest when the option-implied
probability is low and when time to expiry is long. This pattern is consistent with
speculative or sentiment-driven demand for out-of-the-money prediction-market
contracts. It is difficult to reconcile with the alternative hypothesis that the
wedge is driven by liquidity premia or settlement-risk adjustments, which would
predict a different sign and magnitude pattern.

\paragraph{Related literature.}
The paper contributes to three strands of research. The first is the literature on
prediction market efficiency and information aggregation. Early work by
\citet{wolfers2004} and \citet{berg2008} establishes that prediction markets produce
well-calibrated forecasts for political and economic events, while more recent papers
study long-run efficiency properties \citep{rhode2004} and the conditions under
which prices can be interpreted as probabilities \citep{wolfers2006}. Our
contribution is to provide an option-theoretic benchmark
grounded in no-arbitrage pricing rather than ex-post calibration, allowing a sharp
cross-market test independent of realized outcomes.

The second is the literature on cross-market price discovery and limits to
arbitrage. \citet{gromb2010} and \citet{kondor2012} show that capital constraints
prevent full price equalization across markets even when the same security trades
in multiple venues. \citet{duffie2010} documents slow-moving capital as a source
of persistent mispricings in derivatives markets. Our finding of a persistent,
mean-reverting wedge with a four-hour half-life is consistent with these mechanisms:
information and capital move across the Polymarket and Binance ecosystems, but not
instantaneously and not without cost.

The closest antecedent is the literature comparing prediction-market prices with
implied probabilities from financial derivatives \citep{leigh2003, snowberg2010}.
Those papers compare prediction-market prices against alternatives constructed from
equity or options markets for political and macroeconomic events. Our setting
differs in three respects: the underlying event is a financial threshold (Bitcoin
price), the option counterpart is exact rather than approximate, and the data are
aligned at hourly frequency, allowing a time-series characterization of the wedge
rather than a single-point comparison.

The paper also connects to the literature on digitization and financial inclusion.
The expansion of digital prediction platforms has extended financial market access
to participants who would not typically engage with listed derivatives, generating a
new class of digitally native price signals. Understanding whether those signals are
aligned with option-market benchmarks is a prerequisite for evaluating whether
digital platforms genuinely aggregate information or instead reflect segmented,
potentially mispriced demand.

\paragraph{Organization.}
Section~\ref{sec:data} describes the data and contract mapping.
Section~\ref{sec:model} develops the Black--Scholes benchmark, the delta-method
uncertainty propagation, and the hypothesis testing framework.
Section~\ref{sec:results} presents the main results, pooled evidence,
cross-sectional patterns, and persistence analysis.
Section~\ref{sec:backtest} describes the arbitrage-style proxy backtest.
Section~\ref{sec:extension} reports the Deribit extension and Ethereum exercise.
Section~\ref{sec:discussion} interprets the findings and discusses scope conditions.
Section~\ref{sec:conclusion} concludes. Proofs are in the Appendix.

\section{Data}
\label{sec:data}

\subsection{Contract mapping}
\label{sec:mapping}

The key empirical step is matching a prediction-market contract to an options-based
payoff with the same economic meaning. A Polymarket Yes share pays one unit at
maturity $T$ if the terminal Bitcoin level $S_T$ exceeds a threshold $K$.
Abstracting from minor venue-specific settlement details, that payoff is the binary
indicator $\mathbf{1}\{S_T > K\}$. Under a frictionless no-arbitrage benchmark
with constant short rate $r$, its time-$t$ price is
\begin{equation}
e^{-r\tau_t}\mathbb{E}^{\mathbb{Q}}\!\left[\mathbf{1}\{S_T > K\}\mid\mathcal{F}_t\right],
\qquad \tau_t = \frac{T-t}{365}.
\end{equation}
Under Black--Scholes dynamics, this expectation equals the value of a cash-or-nothing
call option, providing a direct bridge from a listed vanilla option price to a
prediction-market benchmark.

\subsection{Sources and sample}
\label{sec:sources}

The analysis combines three public data sources.

\begin{enumerate}
\item \emph{Binance option data.} Observations are taken from the public
\texttt{EOHSummary} archive, which provides historical bid and ask information for
all listed Bitcoin option contracts at end-of-hour frequency \citep{binancepublicdata}.

\item \emph{Binance spot data.} Spot prices are taken from the public hourly bar
archive at the same source \citep{binancepublicdata}.

\item \emph{Polymarket data.} Prediction-market prices are constructed from public
transaction-level trade data and, where available, from historical hourly prices
\citep{polymarketdocs}.
\end{enumerate}

The main specification studies the Polymarket contract ``BTC above \$27{,}000 at
end of September?'' (market id~252196), matched to Binance option
\texttt{BTC-230929-27000-C}. For pooled robustness, I also examine two August~2023
contracts matched to \texttt{BTC-230901-28000-C} and \texttt{BTC-230901-26000-C}.
The pooled exact-data panel contains 287 aligned hourly observations.

\subsection{Variables}

\begin{table}[H]
\centering
\caption{Variable definitions}
\label{tab:vars}
\begin{tabular}{ll}
\toprule
Variable & Definition \\
\midrule
$C_{mkt,t}$    & Binance option mid-price (average of best bid and ask) \\
$S_t$          & Binance spot reference level \\
$K$            & Strike shared by the two contracts \\
$\tau_t$       & Time to Polymarket expiration (years, 365-day convention) \\
$\tau_t^{(c)}$ & Time to Binance option expiration \\
$P_{poly,t}$   & Polymarket Yes price on the unit interval \\
$P_{fair,t}$   & Option-implied risk-neutral binary benchmark \\
$s_{B,t}$      & Relative Binance bid--ask spread \\
$s_P$          & Polymarket spread parameter \\
$f_B,\, f_P$   & Binance and Polymarket fee parameters \\
$D_t$          & Pricing gap $P_{poly,t}-P_{fair,t}$ \\
\bottomrule
\end{tabular}
\end{table}

All Polymarket observations are placed on the Yes-probability scale. A Yes trade
at price $p$ enters as $p$; a No trade at price $p$ enters as $1-p$.

\subsection{Time alignment}

All series are aligned at hourly frequency. Binance spot and option observations are
recorded on the hour. Polymarket observations are merged backward to the latest
available quote or trade at or before the Binance timestamp. Observations with
extremely short maturity are excluded to avoid numerically unstable implied-volatility
estimates.

\section{Model and Inference}
\label{sec:model}

\subsection{Black--Scholes benchmark}
\label{sec:bs}

Under the Black--Scholes model \citep{black1973}, the European call price is
\begin{equation}
C_{BS}(S_t,K,r,\tau,\sigma) = S_t\Phi(d_1)-Ke^{-r\tau}\Phi(d_2),
\end{equation}
with
\begin{equation}
d_1 = \frac{\log(S_t/K)+(r+\sigma^2/2)\tau}{\sigma\sqrt{\tau}}, \qquad
d_2 = d_1-\sigma\sqrt{\tau}.
\end{equation}
The discounted cash-or-nothing call value,
\begin{equation}
\label{eq:pfair}
P_{fair,t}(\sigma) = e^{-r\tau_t}\Phi\!\left(d_{2,t}(\sigma)\right),
\end{equation}
serves as the option-implied risk-neutral benchmark for the prediction-market event
(Proposition~\ref{prop:digital-price}). Black--Scholes is the natural starting point
because any deviation from it is quantifiable and replicable. Stochastic-volatility
corrections would in general widen the measured gap, since Bitcoin smiles are
left-skewed for near-money strikes; the benchmark is therefore conservative with
respect to the sign of the main finding.

\subsection{Implied-volatility inversion}

At each time $t$, implied volatility $\hat\sigma_t$ is found numerically as the
root of $F_t(\sigma) = C_{BS}(S_t,K,r,\tau_t^{(c)},\sigma) - C_{mkt,t}$, then
substituted into \eqref{eq:pfair} to obtain $P_{fair,t}$.

\subsection{Delta-method uncertainty propagation}
\label{sec:delta}

Treating the observed option mid as a latent efficient price plus additive noise,
$C_{mkt,t} = C^\star_t + \varepsilon_t$, a first-order implicit-function expansion
with call vega $\mathcal{V}_t = \partial C_{BS}/\partial\sigma$ gives
\begin{equation}
\operatorname{Var}(P_{fair,t}\mid\mathcal{F}_t)
\approx \left(g'(\hat\sigma_t)\right)^2 \frac{\varsigma_{C,t}^2}{\mathcal{V}_t^2},
\end{equation}
where $g(\sigma)=e^{-r\tau_t}\Phi(d_{2,t}(\sigma))$ and
$g'(\sigma) = -e^{-r\tau_t}\phi(d_2)d_1/\sigma$
(Propositions~\ref{prop:digital-sensitivity}--\ref{prop:iv-delta}). The plug-in
standard error is
\begin{equation}
SE(P_{fair,t}) = \left|g'(\hat\sigma_t)\right|\frac{\hat\varsigma_{C,t}}{\mathcal{V}_t},
\end{equation}
where $\hat\varsigma_{C,t} = \max\!\{(C_{high,t}-C_{low,t})/2,\;0.01|C_{mkt,t}|,
\;10^{-10}\}$ uses the intra-hour high--low range as a conservative scale for
within-hour quote variation.

\subsection{Friction-adjusted comparison band}

The total friction term is
\begin{equation}
TF_t = (f_B+f_P) + \tfrac{1}{2}(s_{B,t}+s_P),
\end{equation}
and the friction-adjusted uncertainty band is
\begin{equation}
CI_{adj,t} = \bigl[P_{fair,t} - 1.96\,SE(P_{fair,t}) - TF_t,\;
             P_{fair,t} + 1.96\,SE(P_{fair,t}) + TF_t\bigr].
\end{equation}

\subsection{Discrepancy process and hypothesis testing}
\label{sec:testing}

The central object is $D_t = P_{poly,t} - P_{fair,t}$ and the null hypothesis is
$H_0:\mathbb{E}[D_t]=0$. The baseline test is the one-sample $t$-statistic
$t_n = \bar{D}_n/(s_D/\sqrt{n})$, supplemented by HAC confidence intervals using
the Newey--West estimator \citep{newey1987} and circular block-bootstrap intervals
with blocks of length $\lfloor n^{1/3}\rfloor$.

\subsection{Dynamic specification}

Persistence is characterized via the AR(1) process
\begin{equation}
D_t = a + \phi D_{t-1} + u_t,
\end{equation}
with implied half-life $h_{1/2} = \log(0.5)/\log(\phi)$ hours and stationarity
assessed via the Augmented Dickey--Fuller test \citep{dickey1979}.

\section{Empirical Results}
\label{sec:results}

\subsection{Main market}

The main specification matches the Polymarket contract ``BTC above \$27{,}000 at
end of September?'' to Binance option \texttt{BTC-230929-27000-C} over 214 aligned
hourly observations. Sample averages are
\[
\overline{C_{mkt}} = 253.45, \quad
\overline{S_t} = 26{,}664.80, \quad
\overline{P_{poly}} = 0.4043, \quad
\overline{P_{fair}} = 0.3484.
\]
The mean discrepancy is $\bar{D} = 0.0558$ ($s_D = 0.1264$), indicating that
Polymarket Yes prices exceeded the Binance-implied risk-neutral benchmark by roughly
5.6 percentage points on average. Table~\ref{tab:main} collects the inferential
results.

\begin{table}[H]
\centering
\caption{Main-market results: BTC above \$27{,}000 at end of September}
\label{tab:main}
\begin{threeparttable}
\begin{tabular}{lc}
\toprule
Statistic & Value \\
\midrule
Aligned observations $n$ & 214 \\
Mean discrepancy $\bar{D}$ & $0.0558$*** \\
Standard deviation $s_D$ & $0.1264$ \\
One-sample $t$-statistic & $6.46$ \\
One-sample $p$-value & $6.86\times10^{-10}$ \\
HAC 95\% CI for $\mathbb{E}[D_t]$ & $[0.0228,\;0.0889]$ \\
Block-bootstrap 95\% CI & $[0.0212,\;0.0911]$ \\
Mean friction term $\overline{TF_t}$ & $0.0427$ \\
Mean $\overline{SE(P_{fair,t})}$ & $0.0458$ \\
Share outside $CI_{adj,t}$ & $0.411$ \\
AR(1) $\hat\phi$ & $0.849\;\pm\;0.051$ \\
Half-life (hours) & $4.2$ \\
ADF statistic & $-3.69$ ($p=0.004$) \\
\bottomrule
\end{tabular}
\begin{tablenotes}[flushleft]
\footnotesize
\item \textit{Notes.} Two-sided one-sample test of $H_0:\mathbb{E}[D_t]=0$.
* $p<0.10$, ** $p<0.05$, *** $p<0.01$. HAC: Newey--West. Bootstrap: circular
blocks of length $\lfloor n^{1/3}\rfloor$.
\end{tablenotes}
\end{threeparttable}
\end{table}

The inference is unambiguous across all three procedures. The AR(1) half-life of
4.2 hours combined with ADF rejection of a unit root ($p=0.004$) characterizes the
wedge as stationary but highly persistent---consistent with slow cross-venue
information transmission rather than transitory noise. Figure~\ref{fig:main_ts}
shows the full time series.

\begin{figure}[H]
\centering
\includegraphics[width=0.88\textwidth]{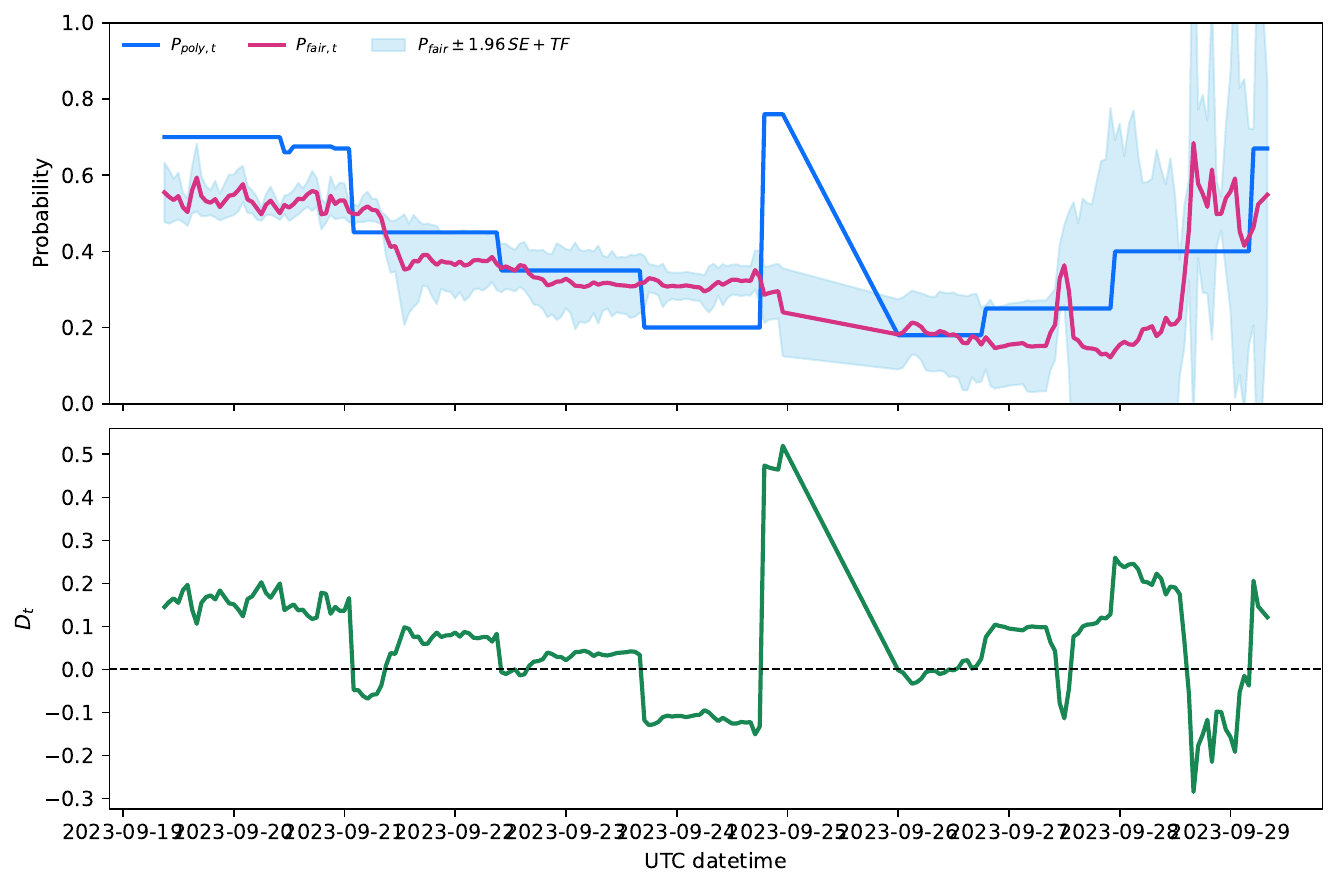}
\caption{Main-market time series. Upper panel: Polymarket Yes price (blue),
Binance-implied benchmark (red), and friction-adjusted band (shaded). Lower panel:
discrepancy process $D_t=P_{poly,t}-P_{fair,t}$.}
\label{fig:main_ts}
\end{figure}

\subsection{Pooled evidence}

Table~\ref{tab:markets} reports results across all three contracts and the pooled
panel.

\begin{table}[H]
\centering
\caption{Market-level and pooled results}
\label{tab:markets}
\begin{threeparttable}
\begin{tabular}{lcccc}
\toprule
Market & $n$ & Mean $D$ & 95\% HAC CI & Share outside $CI_{adj}$ \\
\midrule
BTC $>$ \$28{,}000 end of August     & 37  & $0.0957$*** & $[0.070,\;0.122]$ & 0.405 \\
BTC $>$ \$26{,}000 end of August     & 36  & $0.0745$**  & $[-0.027,\;0.176]$ & 0.028 \\
BTC $>$ \$27{,}000 end of September  & 214 & $0.0558$*** & $[0.023,\;0.089]$  & 0.411 \\
\midrule
Pooled panel                         & 287 & $0.0633$*** & $[0.033,\;0.094]$  & 0.362 \\
\bottomrule
\end{tabular}
\begin{tablenotes}[flushleft]
\footnotesize
\item \textit{Notes.} * $p<0.10$, ** $p<0.05$, *** $p<0.01$.
Pooled one-sample $t=8.24$, $p=6.26\times10^{-15}$.
\end{tablenotes}
\end{threeparttable}
\end{table}

The positive wedge is present in every market. Figure~\ref{fig:pooled_panels}
shows the pooled discrepancy distribution and market-level mean estimates.

\begin{figure}[H]
\centering
\begin{subfigure}[t]{0.47\textwidth}
\centering
\includegraphics[width=\textwidth]{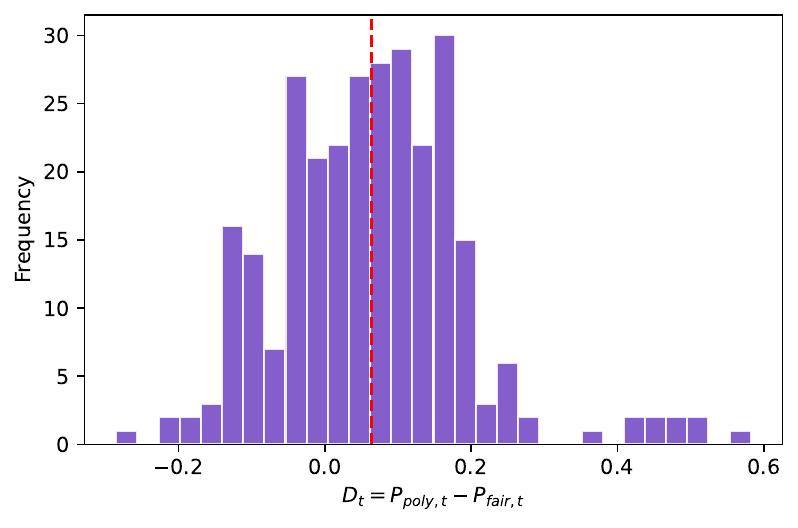}
\caption{Pooled discrepancy distribution.}
\label{fig:pooled_hist}
\end{subfigure}
\hfill
\begin{subfigure}[t]{0.49\textwidth}
\centering
\includegraphics[width=\textwidth]{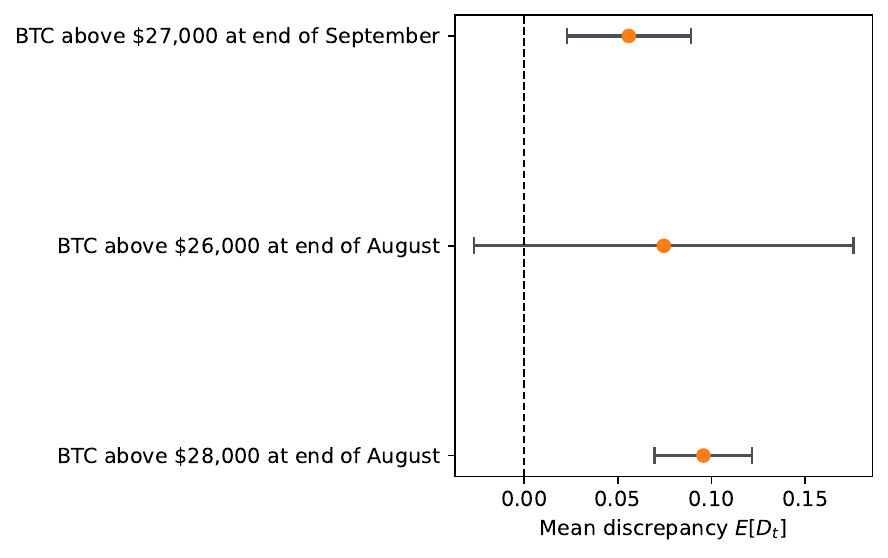}
\caption{Market-level mean estimates with HAC CIs.}
\label{fig:market_means}
\end{subfigure}
\caption{Pooled evidence across the three Bitcoin threshold markets.}
\label{fig:pooled_panels}
\end{figure}

\subsection{Robustness and cross-sectional patterns}
\label{sec:robustness}

The pooled result is not driven by outliers: the median discrepancy is 0.0642, the
10\% trimmed mean is 0.0576, and a sign test strongly rejects a 50--50 split
(68.6\% positive observations, $p=2.45\times10^{-10}$). Table~\ref{tab:robust}
collects these outlier-robust summaries.

\begin{table}[H]
\centering
\caption{Outlier-robust summary statistics}
\label{tab:robust}
\begin{threeparttable}
\begin{tabular}{lcccc}
\toprule
Sample & $n$ & Median $D$ & Trimmed mean & Share$(D_t>0)$ \\
\midrule
BTC $>$ \$28{,}000 end of August    & 37  & $0.100$  & $0.097$ & $1.000$ \\
BTC $>$ \$26{,}000 end of August    & 36  & $-0.036$ & $0.049$ & $0.389$ \\
BTC $>$ \$27{,}000 end of September & 214 & $0.064$  & $0.053$ & $0.682$ \\
Pooled panel                         & 287 & $0.064$  & $0.058$ & $0.686$ \\
\bottomrule
\end{tabular}
\begin{tablenotes}[flushleft]
\footnotesize
\item Trimmed mean drops top and bottom 10\% within each sample.
\end{tablenotes}
\end{threeparttable}
\end{table}

The cross-sectional regression in Table~\ref{tab:reg} reveals the structure of the
wedge: it is largest at low option-implied probabilities, grows with benchmark
uncertainty, and increases with time to expiry. The pattern at low probabilities is
the signature of speculative demand for out-of-the-money prediction-market
contracts, analogous to the favourite--longshot bias in parimutuel markets
\citep{jullien2000}.

\begin{table}[H]
\centering
\caption{Cross-sectional determinants of the pricing gap}
\label{tab:reg}
\begin{threeparttable}
\begin{tabular}{lccc}
\toprule
Variable & Coefficient & Robust SE & $p$-value \\
\midrule
$P_{fair,t}$           & $-0.398$ & $0.063$ & $<0.001$ \\
$SE(P_{fair,t})$       & $\phantom{-}0.039$ & $0.019$ & $0.040$ \\
Time to expiry (hours) & $\phantom{-}0.0008$ & $0.0001$ & $<0.001$ \\
\bottomrule
\end{tabular}
\begin{tablenotes}[flushleft]
\footnotesize
\item Dependent variable: $D_t$. Pooled panel, market fixed effects, $n=287$,
$R^2=0.221$.
\end{tablenotes}
\end{threeparttable}
\end{table}

\begin{figure}[H]
\centering
\begin{subfigure}[t]{0.47\textwidth}
\centering
\includegraphics[width=\textwidth]{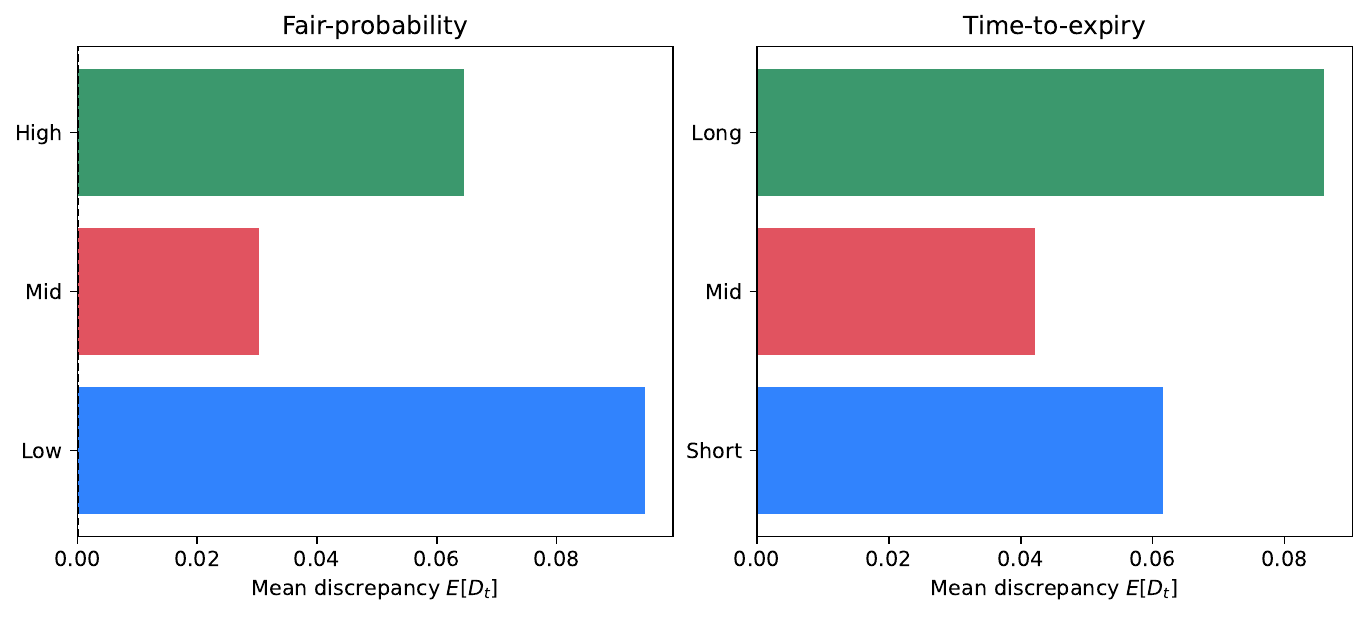}
\caption{Mean gap across terciles.}
\end{subfigure}
\hfill
\begin{subfigure}[t]{0.49\textwidth}
\centering
\includegraphics[width=\textwidth]{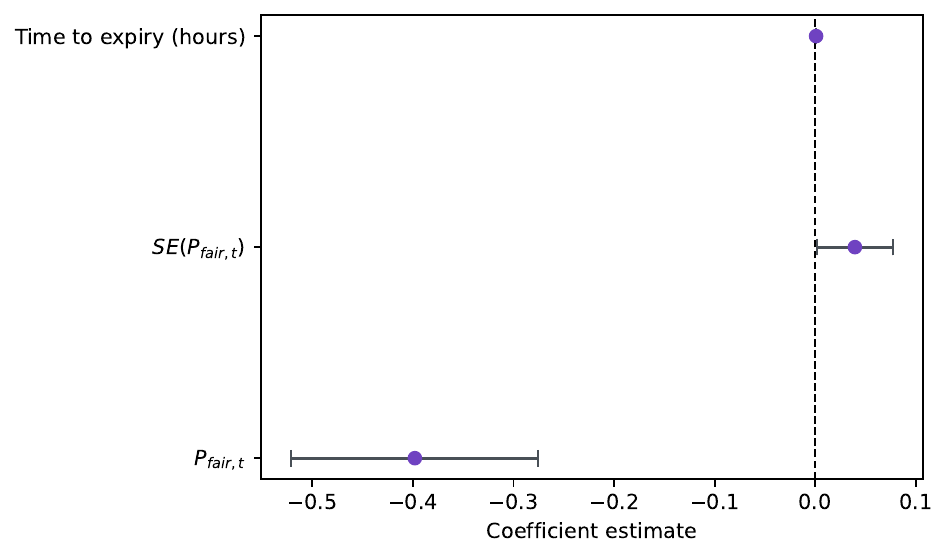}
\caption{Regression coefficients with 95\% CIs.}
\end{subfigure}
\caption{Cross-sectional patterns.}
\label{fig:patterns}
\end{figure}

\section{Arbitrage-Style Proxy Backtest}
\label{sec:backtest}

To test whether the wedge is economically exploitable, I construct a delta-hedged
proxy using the matched Binance call, spot, and a cash account. At entry time $t_0$,
the call quantity $q_0 = \mathcal{V}^D_{t_0}/\mathcal{V}^C_{t_0}$ matches the
digital's vega, and the spot position $x_{t_0} = \Delta^D_{t_0} - q_0\Delta^C_{t_0}$
matches its delta. The spot hedge is rebalanced hourly. Positions open when
$|D_t| > SE(P_{fair,t}) + TF_t$ and close on mean reversion or at expiry. Net PnL
subtracts Polymarket, option, and spot transaction costs. The pooled backtest
generates 16 trades, gross PnL 1.649, and net PnL 1.113 after all cost assumptions.
Net win rate is 69\% and median holding period is 3.5 hours. Pooled net alpha is
0.067 ($t=2.10$, $p=0.053$, HAC CI: $[-0.008,\;0.143]$). Table~\ref{tab:backtest}
gives market-level detail and Figure~\ref{fig:backtest_pnl} plots cumulative PnL.

\begin{table}[H]
\centering
\caption{Arbitrage-style proxy backtest summary}
\label{tab:backtest}
\begin{threeparttable}
\begin{tabular}{lcccccc}
\toprule
Market & Trades & Net PnL & Net $\alpha$ & $p$ & 95\% HAC CI & Med.\ hold (h) \\
\midrule
BTC $>$ \$28{,}000 Aug & 4  & $0.274$   & $0.221$ & $0.103$ & $[0.027,\;0.416]$ & 2.5 \\
BTC $>$ \$26{,}000 Aug & 1  & $-0.345$  & ---     & ---     & ---                & 1.0 \\
BTC $>$ \$27{,}000 Sep & 11 & $1.183$   & $0.018$ & $0.092$ & $[-0.003,\;0.039]$ & 6.0 \\
\midrule
Pooled                  & 16 & $1.113$   & $0.067$ & $0.053$ & $[-0.008,\;0.143]$ & 3.5 \\
\bottomrule
\end{tabular}
\begin{tablenotes}[flushleft]
\footnotesize
\item Alpha = mean net return relative to deployed notional.
\end{tablenotes}
\end{threeparttable}
\end{table}

\begin{figure}[H]
\centering
\includegraphics[width=0.72\textwidth]{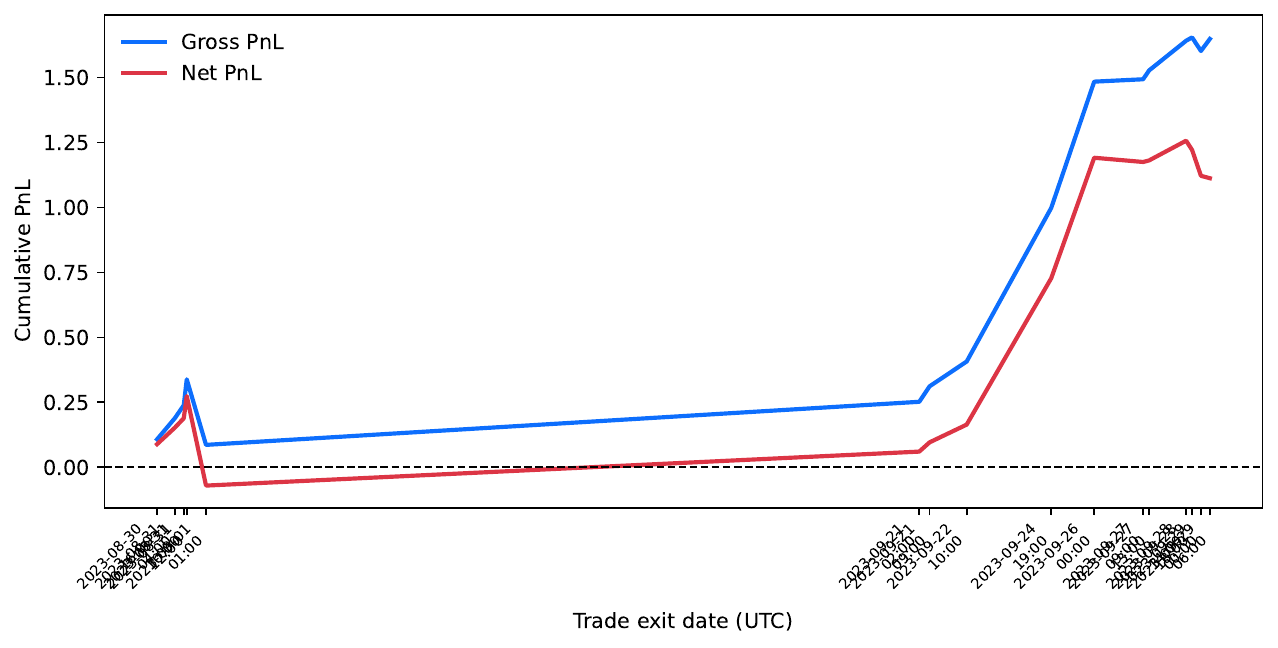}
\caption{Cumulative gross (blue) and net (red) PnL across trade exits.}
\label{fig:backtest_pnl}
\end{figure}

\section{Benchmark Extension: Deribit and Ethereum}
\label{sec:extension}

\subsection{Bitcoin against Deribit}

Re-running the same three Bitcoin threshold markets against Deribit call options
confirms the direction but reveals cross-market heterogeneity. In the
September~\$27{,}000 contract, the mean discrepancy is 0.1248 over 645 observations
(HAC CI: $[0.075,\;0.174]$). The pooled three-market Deribit panel gives a mean
of 0.1105 over 2{,}585 observations (HAC CI: $[0.074,\;0.147]$). Table~\ref{tab:deribit}
compares both benchmarks.

\begin{table}[H]
\centering
\caption{Binance vs.\ Deribit benchmark comparison}
\label{tab:deribit}
\begin{threeparttable}
\begin{tabular}{llcccc}
\toprule
Sample & Benchmark & $n$ & Mean $D$ & 95\% HAC CI & $p$-value \\
\midrule
Sep \$27{,}000 & Binance & 214   & 0.0558 & $[0.023,\;0.089]$ & $6.9\times10^{-10}$ \\
Sep \$27{,}000 & Deribit & 645   & 0.1248 & $[0.075,\;0.174]$ & $2.5\times10^{-32}$ \\
\midrule
Pooled & Binance & 287    & 0.0633 & $[0.033,\;0.094]$ & $6.3\times10^{-15}$ \\
Pooled & Deribit & 2{,}585 & 0.1105 & $[0.074,\;0.147]$ & $5.8\times10^{-64}$ \\
\bottomrule
\end{tabular}
\end{threeparttable}
\end{table}

\subsection{Ethereum exercise}

A smaller exercise matches four Polymarket ``ETH above \$1{,}700'' contracts from
February--March 2023 to Deribit ETH calls. The pooled four-market mean discrepancy
is 0.0129 over 2{,}737 observations (HAC CI: $[-0.000,\;0.026]$), with clear
instability across contracts: the three February markets show slightly negative
discrepancies, while the March~24 market is strongly positive at 0.1899
(HAC CI: $[0.145,\;0.234]$). The Ethereum exercise is supplementary evidence rather
than a second main result.

\section{Discussion}
\label{sec:discussion}

\subsection{Scope conditions}

The benchmark is model-based. Black--Scholes is a clean, replicable standard but
not a literal description of cryptocurrency dynamics. The result should be read as
a deviation from the standard option-implied risk-neutral benchmark rather than a
statement about physical probabilities. Stochastic-volatility corrections would in
general widen the gap, making the Black--Scholes benchmark conservative. The
contract mapping is close but not exact: Polymarket resolves under oracle procedures
while the option benchmark follows exchange convention. The sample is intentionally
narrow---three Binance-compatible contracts---which rules out mapping ambiguity but
limits external validity.

\subsection{Interpretation}

Three mechanisms are consistent with the evidence. First, \textit{demand-side
overpricing}: the strong negative relationship between the wedge and the
option-implied probability mirrors the favourite--longshot bias in parimutuel markets
\citep{thaler1988, jullien2000}. Retail participants on digital prediction platforms
over-weight unlikely outcomes, bidding prices above the risk-neutral consensus.
Second, \textit{segmented capital and slow arbitrage}: the four-hour half-life is
consistent with limits to arbitrage \citep{gromb2010, duffie2010}; capital that
could equalize prices is either absent or costly to deploy across two different
platform ecosystems. Third, \textit{risk-neutral measure differences}: if the two
venues embed systematically different risk adjustments, a persistent gap could
survive in frictionless equilibrium---though this requires a structural model of
preference heterogeneity not developed here.

\section{Conclusion}
\label{sec:conclusion}

This paper provides the first option-implied benchmark test of prediction-market
pricing for cryptocurrency threshold contracts. In the main September~2023 Bitcoin
contract, the mean discrepancy is 5.6 percentage points across 214 hourly
observations ($p<10^{-9}$). In the pooled three-market panel, it is 6.3 percentage
points across 287 observations ($p<10^{-14}$). The gap is largest at low
option-implied probabilities and long maturities, robust to outlier removal and
sign tests, and persistent with a four-hour half-life. A delta-hedged proxy backtest
generates positive net PnL after conservative costs. The Deribit extension
preserves the sign and economic magnitude of the finding; the Ethereum exercise is
more mixed.

The central takeaway is that digital fragmentation of financial markets generates
systematic, persistent pricing wedges even for economically identical payoffs.
Two digitally native platforms---a centralized option exchange and a blockchain
prediction market---do not agree on the value of the same binary event, and the
structure of the disagreement points toward demand-side overpricing on the
prediction-market side. This has implications for how prediction-market prices
should be interpreted as probability estimates, and for how researchers understand
the relationship between digital prediction platforms and centralized derivatives
markets.

\appendix
\section{Analytical Derivations}
\label{sec:appendix}

Throughout this appendix, $\phi(\cdot)$ and $\Phi(\cdot)$ denote the standard
normal density and cumulative distribution function, respectively. We write
$\tau = T - t > 0$ for time to maturity.

\begin{proposition}[Risk-neutral valuation of the binary payoff]
\label{prop:digital-price}
Let $S$ follow geometric Brownian motion under the risk-neutral measure
$\mathbb{Q}$,
\[
\frac{dS_u}{S_u} = r\,du + \sigma\,dW_u^{\mathbb{Q}}, \qquad u \in [t,T],
\]
with constant drift $r\geq 0$ and volatility $\sigma > 0$. Then the time-$t$
no-arbitrage price of the unit payoff $\mathbf{1}\{S_T > K\}$ is
\[
P_{fair,t} = e^{-r\tau}\,\Phi(d_2),
\]
where
\[
d_2 = \frac{\log(S_t/K) + (r - \tfrac{1}{2}\sigma^2)\tau}{\sigma\sqrt{\tau}}.
\]
\end{proposition}

\begin{proof}
\textit{Step 1: Log-normal distribution of $S_T$.}
Applying It\^{o}'s lemma to $\log S_u$ and integrating from $t$ to $T$ gives
\[
\log S_T = \log S_t + \left(r - \tfrac{1}{2}\sigma^2\right)\tau
           + \sigma\int_t^T dW_u^{\mathbb{Q}}.
\]
Since $\int_t^T dW_u^{\mathbb{Q}} \sim N(0,\tau)$ under $\mathbb{Q}$, we can write
\[
\log S_T = \log S_t + \left(r - \tfrac{1}{2}\sigma^2\right)\tau
           + \sigma\sqrt{\tau}\,Z, \qquad Z \sim N(0,1) \text{ under } \mathbb{Q}.
\]

\textit{Step 2: Risk-neutral event probability.}
The event $\{S_T > K\}$ is equivalent to
\[
\left\{Z > \frac{\log(K/S_t) - (r-\tfrac{1}{2}\sigma^2)\tau}{\sigma\sqrt{\tau}}\right\}
= \{Z > -d_2\}.
\]
Since $Z \sim N(0,1)$ under $\mathbb{Q}$ and $N(0,1)$ is symmetric,
\[
\mathbb{Q}(S_T > K \mid \mathcal{F}_t)
= \mathbb{Q}(Z > -d_2)
= 1 - \Phi(-d_2)
= \Phi(d_2).
\]

\textit{Step 3: No-arbitrage pricing.}
Under the risk-neutral measure, the time-$t$ price of any payoff equals its
discounted expectation under $\mathbb{Q}$. Therefore
\[
P_{fair,t}
= e^{-r\tau}\,\mathbb{E}^{\mathbb{Q}}\!\bigl[\mathbf{1}\{S_T > K\}\mid\mathcal{F}_t\bigr]
= e^{-r\tau}\,\mathbb{Q}(S_T > K \mid \mathcal{F}_t)
= e^{-r\tau}\,\Phi(d_2). \qedhere
\]
\end{proof}

\begin{proposition}[Black--Scholes vega]
\label{prop:vega}
Let $C_{BS}(S_t,K,r,\tau,\sigma) = S_t\Phi(d_1) - Ke^{-r\tau}\Phi(d_2)$ denote
the Black--Scholes call price, with $d_1 = [\log(S_t/K)+(r+\tfrac{1}{2}\sigma^2)\tau]
/(\sigma\sqrt{\tau})$ and $d_2 = d_1 - \sigma\sqrt{\tau}$. Then
\[
\mathcal{V}(S_t,K,r,\tau,\sigma)
\;:=\;
\frac{\partial C_{BS}}{\partial\sigma}
= S_t\,\phi(d_1)\,\sqrt{\tau}.
\]
\end{proposition}

\begin{proof}
\textit{Step 1: Auxiliary identity.}
We first establish that $S_t\phi(d_1) = Ke^{-r\tau}\phi(d_2)$. Compute
\[
d_1^2 - d_2^2 = (d_1+d_2)(d_1-d_2).
\]
Since $d_1 - d_2 = \sigma\sqrt{\tau}$ and
$d_1 + d_2 = 2[\log(S_t/K) + r\tau]/(\sigma\sqrt{\tau})$, we obtain
\[
d_1^2 - d_2^2
= \frac{2[\log(S_t/K)+r\tau]}{\sigma\sqrt{\tau}}\cdot\sigma\sqrt{\tau}
= 2\bigl[\log(S_t/K)+r\tau\bigr].
\]
Therefore
\[
\frac{\phi(d_1)}{\phi(d_2)}
= \exp\!\left(-\frac{d_1^2-d_2^2}{2}\right)
= \exp\!\bigl(-\log(S_t/K)-r\tau\bigr)
= \frac{K}{S_t}\,e^{-r\tau},
\]
which gives $S_t\phi(d_1) = Ke^{-r\tau}\phi(d_2)$.

\textit{Step 2: Differentiation.}
Differentiating $C_{BS}$ with respect to $\sigma$:
\[
\frac{\partial C_{BS}}{\partial\sigma}
= S_t\,\phi(d_1)\,\frac{\partial d_1}{\partial\sigma}
- Ke^{-r\tau}\,\phi(d_2)\,\frac{\partial d_2}{\partial\sigma}.
\]
Applying the identity from Step 1, both terms share the factor $S_t\phi(d_1)$:
\[
\frac{\partial C_{BS}}{\partial\sigma}
= S_t\,\phi(d_1)\left(\frac{\partial d_1}{\partial\sigma}
  - \frac{\partial d_2}{\partial\sigma}\right).
\]
Since $d_1 - d_2 = \sigma\sqrt{\tau}$, differentiating with respect to $\sigma$
gives $\partial d_1/\partial\sigma - \partial d_2/\partial\sigma = \sqrt{\tau}$.
Substituting yields $\partial C_{BS}/\partial\sigma = S_t\phi(d_1)\sqrt{\tau}$.
\end{proof}

\begin{proposition}[Volatility sensitivity of the binary price]
\label{prop:digital-sensitivity}
Let $g(\sigma) = e^{-r\tau}\Phi(d_2(\sigma))$, where $d_2(\sigma)$ is defined
in Proposition~\ref{prop:digital-price}. Then
\[
g'(\sigma) = -\,e^{-r\tau}\,\phi(d_2)\,\frac{d_1}{\sigma}.
\]
\end{proposition}

\begin{proof}
\textit{Step 1: Chain rule.}
By the chain rule,
\[
g'(\sigma) = e^{-r\tau}\,\phi(d_2(\sigma))\,\frac{\partial d_2}{\partial\sigma}.
\]

\textit{Step 2: Derivative of $d_2$ with respect to $\sigma$.}
Write $A := \log(S_t/K) + r\tau$ (a constant in $\sigma$). Then
\[
d_2(\sigma)
= \frac{A + \tfrac{1}{2}\sigma^2\tau}{\sigma\sqrt{\tau}} - \sigma\sqrt{\tau}
= \frac{A}{\sigma\sqrt{\tau}} + \frac{\sigma\sqrt{\tau}}{2} - \sigma\sqrt{\tau}
= \frac{A}{\sigma\sqrt{\tau}} - \frac{\sigma\sqrt{\tau}}{2}.
\]
Differentiating term by term:
\[
\frac{\partial d_2}{\partial\sigma}
= -\frac{A}{\sigma^2\sqrt{\tau}} - \frac{\sqrt{\tau}}{2}.
\]

\textit{Step 3: Relating the derivative to $d_1/\sigma$.}
From the definition of $d_1$,
\[
\frac{d_1}{\sigma}
= \frac{\log(S_t/K)+(r+\tfrac{1}{2}\sigma^2)\tau}{\sigma^2\sqrt{\tau}}
= \frac{A}{\sigma^2\sqrt{\tau}} + \frac{\sqrt{\tau}}{2}.
\]
Comparing with Step~2,
\[
\frac{\partial d_2}{\partial\sigma}
= -\left(\frac{A}{\sigma^2\sqrt{\tau}} + \frac{\sqrt{\tau}}{2}\right)
= -\frac{d_1}{\sigma}.
\]

\textit{Step 4: Conclusion.}
Substituting back into Step~1:
\[
g'(\sigma) = e^{-r\tau}\,\phi(d_2)\cdot\left(-\frac{d_1}{\sigma}\right)
           = -\,e^{-r\tau}\,\phi(d_2)\,\frac{d_1}{\sigma}. \qedhere
\]
\end{proof}

\begin{proposition}[Delta-method variance of the binary benchmark]
\label{prop:iv-delta}
Let $F_t(\sigma) := C_{BS}(S_t,K,r,\tau^{(c)}_t,\sigma) - C_{mkt,t}$, and
suppose $\hat\sigma_t$ is the unique root of $F_t$, with call vega
$\mathcal{V}_t := \partial C_{BS}/\partial\sigma \big|_{\sigma=\hat\sigma_t} > 0$.
Write $C_{mkt,t} = C^\star_t + \varepsilon_t$ where
$\mathbb{E}[\varepsilon_t\mid\mathcal{F}_t] = 0$ and
$\operatorname{Var}(\varepsilon_t\mid\mathcal{F}_t) = \varsigma_{C,t}^2 < \infty$.
If $P_{fair,t} = g(\hat\sigma_t)$ with $g$ as in
Proposition~\ref{prop:digital-sensitivity}, then to first order in $\varepsilon_t$,
\[
\operatorname{Var}(P_{fair,t}\mid\mathcal{F}_t)
\approx
\frac{\bigl(g'(\hat\sigma_t)\bigr)^2\,\varsigma_{C,t}^2}{\mathcal{V}_t^2}.
\]
\end{proposition}

\begin{proof}
\textit{Step 1: Implicit differentiation.}
Define the map $C \mapsto \sigma^*(C)$ by
$F_t(\sigma^*(C)) = C_{BS}(\sigma^*(C)) - C = 0$.
Since $\mathcal{V}_t \neq 0$, the implicit function theorem guarantees that
$\sigma^*(C)$ is continuously differentiable in a neighbourhood of $C^\star_t$,
with derivative
\[
\frac{d\sigma^*}{dC}\bigg|_{C=C^\star_t}
= \frac{1}{\mathcal{V}_t}.
\]

\textit{Step 2: First-order expansion of $\hat\sigma_t$.}
Because $\hat\sigma_t = \sigma^*(C_{mkt,t}) = \sigma^*(C^\star_t + \varepsilon_t)$,
a first-order Taylor expansion gives
\[
\hat\sigma_t - \sigma^\star_t
= \frac{\varepsilon_t}{\mathcal{V}_t} + O(\varepsilon_t^2).
\]
Taking conditional variance and discarding the $O(\varepsilon_t^2)$ remainder:
\[
\operatorname{Var}(\hat\sigma_t\mid\mathcal{F}_t)
\approx \frac{\varsigma_{C,t}^2}{\mathcal{V}_t^2}.
\]

\textit{Step 3: Delta method for $g(\hat\sigma_t)$.}
Since $g$ is continuously differentiable, a first-order expansion of
$P_{fair,t} = g(\hat\sigma_t)$ around $\sigma^\star_t$ gives
\[
P_{fair,t} - g(\sigma^\star_t)
= g'(\sigma^\star_t)\,(\hat\sigma_t - \sigma^\star_t) + O(\varepsilon_t^2).
\]
Taking conditional variance and replacing $\sigma^\star_t$ with the plug-in
$\hat\sigma_t$:
\[
\operatorname{Var}(P_{fair,t}\mid\mathcal{F}_t)
\approx \bigl(g'(\hat\sigma_t)\bigr)^2\,\operatorname{Var}(\hat\sigma_t\mid\mathcal{F}_t)
= \frac{\bigl(g'(\hat\sigma_t)\bigr)^2\,\varsigma_{C,t}^2}{\mathcal{V}_t^2}. \qedhere
\]
\end{proof}

\begin{proposition}[Tick-noise contribution to discrepancy variance]
\label{prop:tick-noise}
Suppose the observed prediction-market price satisfies
$P^{obs}_{poly,t} = P^\star_{poly,t} + \varepsilon_t$, where
$\varepsilon_t \sim \mathrm{Unif}[-\Delta_P/2,\, \Delta_P/2]$ is independent
of the option-implied pricing error. Then
\[
\mathbb{E}[\varepsilon_t] = 0, \qquad
\operatorname{Var}(\varepsilon_t) = \frac{\Delta_P^2}{12},
\]
and the discrepancy variance satisfies the first-order approximation
\[
\operatorname{Var}(D_t) \approx \operatorname{Var}(P_{fair,t}) + \frac{\Delta_P^2}{12}.
\]
\end{proposition}

\begin{proof}
\textit{Step 1: Moments of the uniform distribution.}
Let $a = \Delta_P/2$, so $\varepsilon_t \sim \mathrm{Unif}[-a,a]$ with
density $f(\varepsilon) = 1/(2a)$ on $[-a,a]$. Then
\[
\mathbb{E}[\varepsilon_t]
= \int_{-a}^{a} \varepsilon\,\frac{1}{2a}\,d\varepsilon
= \frac{1}{2a}\left[\frac{\varepsilon^2}{2}\right]_{-a}^{a}
= \frac{1}{2a}\cdot 0
= 0.
\]
Since $\mathbb{E}[\varepsilon_t]=0$,
\[
\operatorname{Var}(\varepsilon_t)
= \mathbb{E}[\varepsilon_t^2]
= \int_{-a}^{a} \varepsilon^2\,\frac{1}{2a}\,d\varepsilon
= \frac{1}{2a}\left[\frac{\varepsilon^3}{3}\right]_{-a}^{a}
= \frac{1}{2a}\cdot\frac{2a^3}{3}
= \frac{a^2}{3}
= \frac{\Delta_P^2}{12}.
\]

\textit{Step 2: Variance decomposition.}
The discrepancy is
\[
D_t = P^{obs}_{poly,t} - P_{fair,t}
= \bigl(P^\star_{poly,t} - P_{fair,t}\bigr) + \varepsilon_t.
\]
By independence of $\varepsilon_t$ from $P^\star_{poly,t} - P_{fair,t}$,
\[
\operatorname{Var}(D_t)
= \operatorname{Var}(P^\star_{poly,t} - P_{fair,t})
  + \operatorname{Var}(\varepsilon_t).
\]

\textit{Step 3: Approximation.}
At the observation level, the dominant source of variation in
$P^\star_{poly,t} - P_{fair,t}$ is the estimation uncertainty in $P_{fair,t}$
(Proposition~\ref{prop:iv-delta}), since the latent continuous prediction-market
price $P^\star_{poly,t}$ is treated as locally fixed conditional on $\mathcal{F}_t$.
Under this approximation,
\[
\operatorname{Var}(P^\star_{poly,t} - P_{fair,t})
\approx \operatorname{Var}(P_{fair,t}),
\]
and therefore
\[
\operatorname{Var}(D_t) \approx \operatorname{Var}(P_{fair,t}) + \frac{\Delta_P^2}{12}. \qedhere
\]
\end{proof}

\bibliographystyle{plainnat}
\bibliography{refs}

\end{document}